\documentclass[a4paper,11pt]{article}

\usepackage{fullpage}
\usepackage{times}
\usepackage{placeins}

\usepackage[utf8]{inputenc}
\usepackage{url}
\usepackage{xcolor}

\usepackage{amsmath}
\usepackage{amsthm}
\usepackage{amssymb}
\usepackage{bbm}
\usepackage{dsfont}
\usepackage{booktabs}

\usepackage[linesnumbered,ruled,vlined]{algorithm2e}
\usepackage{algpseudocode}

\SetCommentSty{mycommfont}

\usepackage{enumitem}
\usepackage{enumerate}

\usepackage[nameinlink]{cleveref}

\usepackage{todonotes}
\usepackage{tabularx}
\usepackage{colortbl}
\usepackage{multirow}
\usepackage{float}

\usepackage[small]{caption}
\usepackage{graphicx}

\usepackage{tikz}
\usetikzlibrary{positioning,arrows.meta}

\usepackage{natbib}
\usepackage{authblk}

\newtheorem{theorem}{Theorem}[section]

\newtheorem{proposition}[theorem]{Proposition}
\newtheorem{lemma}[theorem]{Lemma}

\theoremstyle{definition}
\newtheorem{definition}[theorem]{Definition}

\usepackage{wrapfig}

\allowdisplaybreaks

\title{\bf On Multi-Level Apportionment}

\author[1]{Ulrike Schmidt-Kraepelin}
\author[2]{Warut Suksompong}
\author[2]{Steven Wijaya}

\affil[1]{TU Eindhoven, The Netherlands}
\affil[2]{National University of Singapore, Singapore}

\date{\vspace{-10mm}}

\begin{document}

\maketitle

\begin{abstract}
Apportionment refers to the well-studied problem of allocating legislative seats among parties or groups with different entitlements. We present a multi-level generalization of apportionment where the groups form a hierarchical structure, which gives rise to stronger versions of the upper and lower quota notions. We show that running Adams' method level-by-level satisfies upper quota, while running Jefferson's method or the quota method level-by-level guarantees lower quota. Moreover, we prove that both quota notions can always be fulfilled simultaneously. 
\end{abstract}

\section{Introduction}
\label{sec:intro}

The problem of \emph{apportionment} in mathematics and politics involves distributing a set of identical resources, often legislative seats, among a number of parties, regions, or groups with differing entitlements \citep{BalinskiYo01,Pukelsheim14}.
These entitlements are typically derived from the populations of the regions or groups, or the votes that the parties receive in an election.
The difficulty of apportionment stems from the fact that only integer numbers of seats can be apportioned whereas the proportional shares according to the entitlements are rarely integers.
In addition to its long history, apportionment continues to receive significant interest from researchers nowadays \citep{Brams2019,cembrano2022multidimensional,Brill2024,Mathieu2024}.

Since the exact proportional share of each group---also called the \emph{quota}---cannot always be achieved, two natural and well-studied notions in apportionment are \emph{upper quota} and \emph{lower quota}.\footnote{\citet{Koriyama2013} suggested other reasons why one may want to deviate from strict proportionality.}
Upper quota dictates that each group should receive at most its quota rounded up, while lower quota requires the group to obtain at least its quota rounded down.
Prominent apportionment methods that fulfill at least one of the quota notions include \emph{Adams' method}, \emph{Jefferson's method}, and the \emph{quota method}.
Each of these methods can be formulated as an iterative procedure that allocates one seat at a time, with the recipient of each subsequent seat determined by the allocation thus far.
Adams' method guarantees upper quota, while Jefferson's method ensures lower quota.
The quota method is an adaptation of Jefferson's method that checks against upper quota violations, and thereby achieves both quota notions.
Due to their iterative nature, all three methods satisfy an intuitive property called \emph{house monotonicity}, which means that the number of seats that each group receives does not decrease as the total number of seats increases.\footnote{For a detailed overview of the theory of apportionment, we refer to Appendix~A of the book by \citet{BalinskiYo01}.
In particular, the concept of satisfying upper and lower quota is called ``staying within the quota'' and discussed in Section~7 of Appendix~A of that book.}

In this paper, we study a generalized version of apportionment where instead of the groups being on a single level, they form an arbitrary hierarchical structure.
This is applicable, for instance, when seats are divided among cities in a country and then among districts in each city.
Ideally, we would like the quota notions to be satisfied not only for cities with respect to the country or for districts with respect to cities, but also for districts with respect to the country.
Another example is when university personnel need to be allotted to faculties and further to departments within each faculty.
We represent the hierarchical structure by a tree and each group by a node in the tree, and call the resulting problem \emph{multi-level apportionment}.
Note that canonical apportionment corresponds to the special case where the tree consists of a single level (besides its root).
The system of \emph{apparentment}, also known as \emph{list alliances} \citep{Leutgab2009,Bochsler2010,Karpov2015}, can be seen as two-level apportionment, as it allows groups to form a coalition and receive seats collectively in the main apportionment process, before distributing these seats among themselves.
However, research on apparentment has primarily focused on the incentives for groups to form such coalitions, rather than on extending the quota notions and house monotonicity to a multi-level framework.

To illustrate the challenges that multi-level apportionment brings, consider the two trees in \Cref{fig:intro-tree}, and assume that there are six seats to be allocated.
In the single-level tree on the left, each leaf node has a weight---representing its entitlement---of $1/4$ relative to the root.
Giving two seats each to nodes $1$ and $2$ and one seat each to nodes $3$ and $4$ is sufficient to satisfy both upper and lower quota.
On the other hand, in the multi-level tree on the right, each leaf node again has weight $1/4$ relative to the root, but nodes $1$ and $2$ belong to one group with weight $1/2$ (represented by node~$5$) while nodes $3$ and $4$ belong to another such group (represented by node~$6$).
If we use the same distribution of seats to nodes $1,2,3,4$ as before, node~$5$ will receive four seats while node~$6$ will receive only two.
As a result, the upper quota of three for node~$5$ and the lower quota of three for node~$6$ relative to the root are both violated.

\begin{figure}[t]
\centering
\begin{tikzpicture}[
    node/.style={circle, draw, minimum size=0.75cm},
    edge from parent/.style={draw},
    level 1/.style={level distance=1.5cm},
    level 2/.style={sibling distance=1.5cm, level distance=1.5cm},
  ]

  % Left tree
  \node[node] (A) {$0$}
    child {node[node] {$1$} edge from parent node[left, xshift=-0.3cm] {$\frac{1}{4}$}}
    child {node[node] {$2$} edge from parent node[left] {$\frac{1}{4}$}}
    child {node[node] {$3$} edge from parent node[right] {$\frac{1}{4}$}}
    child {node[node] {$4$} edge from parent node[right, xshift=0.3cm] {$\frac{1}{4}$}};
    
  % Right tree
  \node[node] (B) at (7, 0) {$0$} [sibling distance=3cm]
    child {node[node] {$5$}
      child {node[node] {$1$} edge from parent node[left] {$\frac{1}{2}$}}
      child {node[node] {$2$} edge from parent node[right] {$\frac{1}{2}$}}
      edge from parent node[left, xshift=-0.15cm] {$\frac{1}{2}$}}
    child {node[node] {$6$}
      child {node[node] {$3$} edge from parent node[left] {$\frac{1}{2}$}}
      child {node[node] {$4$} edge from parent node[right] {$\frac{1}{2}$}}
      edge from parent node[right, xshift=0.15cm] {$\frac{1}{2}$}};
\end{tikzpicture}
\caption{Example tree structures of single-level apportionment (left) and multi-level apportionment (right).}
\label{fig:intro-tree}
\end{figure}
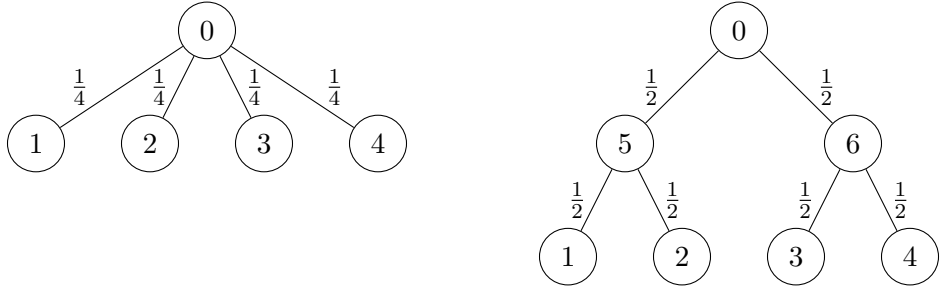

In spite of this apparent difficulty, we prove in \Cref{sec:within-quota} that for any multi-level apportionment instance, it is possible to distribute the seats so as to fulfill both upper and lower quota. 
Specifically, our notions of upper and lower quota require each node in the tree to satisfy quotas derived by comparing its seat allocation to that of \emph{every} of its ancestors.
This means that strong quota guarantees can be made regardless of the hierarchical structure and entitlements.
In \Cref{sec:adams-jefferson}, we explore running Adams' and Jefferson's methods level-by-level and show that the resulting rules ensure upper quota and lower quota, respectively.
As a consequence, either quota notion can be satisfied in conjunction with house monotonicity.
In \Cref{sec:quota-method}, we demonstrate that while an analogous extension of the quota method satisfies lower quota, the single-level check typically performed by the quota method is insufficient to guard against upper quota violations in the multi-level setting.
On the other hand, modifying the check to span multiple levels, while guaranteeing upper quota, can result in a violation of lower quota.
Finally, we conduct an experimental analysis in \Cref{sec:experiment}.

\section{Preliminaries}
\label{sec:prelim}

An instance of multi-level apportionment consists of a rooted tree with $n$ nodes for some positive integer~$n$, entitlements (to be defined in the next paragraph), and a non-negative number of seats $h$ to be allocated.
Assume that the nodes are numbered $0,1,\dots,n-1$, with the root node having number~$0$.

For each node~$i$, denote by $W_i\in (0,1]$ its \emph{entitlement} relative to its parent (for the root node~$0$, we let $W_0 = 1$). We assume that entitlements are normalized, i.e., $\sum_{c\in C(i)} W_c = 1$ holds for every non-leaf node $i$, where $C(i)$ denotes the set of $i$'s children. Moreover, for any node $i$, we let $A(i)$ be the set of $i$'s ancestors. Formally, $A(0) = \{0\}$,\footnote{This choice will be more convenient for our purposes than defining $A(0) = \emptyset$.} and recursively down the tree, for each node $i$ we define $A(i) = \{j\} \cup A(j)$, where $j$ is the parent of $i$. 

The $h$ seats need to be allocated among the nodes in the tree.
Let $V_i$ be the number of seats allocated to node~$i$.
The root node always receives all $h$ seats, and the number of seats allocated to any non-leaf node must be equal to the total number of seats allocated to its children.
That is, $V_0 = h$ and $V_i = \sum_{c \in C(i)} V_c$ for each non-leaf node~$i$.
A \emph{(multi-level) apportionment rule} maps any (multi-level) apportionment instance to an allocation of the seats.

For each node~$i$, denote by $R_i$ the entitlement of node~$i$ relative to the root node, i.e., $R_i := \prod_{a\in A(i)\cup \{i\}} W_a$.
\begin{definition}
The \emph{lower quota} of node~$i$ is defined as $LQ_i := \max_{a \in A(i)} \lfloor \frac{R_i}{R_a} \cdot V_a \rfloor$, and the \emph{upper quota} as $UQ_i := \min_{a \in A(i)} \lceil \frac{R_i}{R_a} \cdot V_a \rceil$. 
\end{definition}
Note that both lower and upper quota capture the quota requirements of a node with respect to every of its ancestors in the tree.
An allocation of seats is said to fulfill lower (resp.,~upper) quota if $V_i\ge LQ_i$ (resp., $V_i\le UQ_i$) for every~$i$.\footnote{One could define lower and upper quota with respect to only the root node, i.e., $\lfloor R_i\cdot h\rfloor$ and $\lceil R_i\cdot h\rceil$, respectively.
This leads to weaker quota notions, so our positive results carry over to these notions as well.
Also, our quota notions are stronger than those defined by considering each node only with respect to its parent.
For example, in the instance shown in \Cref{fig:quota2-lq-tree}, the allocation described in the proof of \Cref{prop:quota2-lower-quota} satisfies the latter lower quota notion but not ours.
\label{foot:relaxation}}
An apportionment rule satisfies lower (resp.,~upper) quota if, for every instance, it returns a seat allocation that fulfills lower (resp.,~upper) quota.
As an example, for the right tree in \Cref{fig:intro-tree} with $h = 6$, consider an allocation with $V_0 = 6$, $V_1 = V_2 = 2$, $V_3 = V_4 = 1$, $V_5 = 4$, and $V_6 = 2$.
This allocation violates lower quota because $LQ_6 = \lfloor \frac{1}{2}\cdot 6\rfloor = 3 > V_6$; it also violates upper quota because $UQ_5 = \lceil \frac{1}{2}\cdot 6\rceil = 3 <  V_5$.

Finally, an apportionment rule satisfies \emph{house monotonicity} if it has the following property: for any instance, if the total number of seats to be allocated increases from $h$ to $h+1$, then every node receives no fewer seats than before.
A common way to satisfy house monotonicity is to iteratively allocate one seat at a time.
For such an iterative method, we denote by $V_i^h$ the number of seats allocated to node~$i$ after $h$~seats have been allocated in total.

\section{Staying Within the Quotas}
\label{sec:within-quota}

In single-level apportionment, it is easy to see that there always exists a seat allocation that fulfills both upper and lower quota.\footnote{For example, \emph{Hamilton's method} first gives each node its lower quota, then assigns the remaining seats to the nodes with the highest fractional remainders of their quotas until all seats have been allocated.}
The situation is less clear in multi-level apportionment, however, as the internal nodes can introduce several additional constraints on the permissible allocation.
Nevertheless, we show that for any instance, there exists an allocation that respects both quotas, thereby demonstrating that staying within the quotas is always possible regardless of the hierarchical structure.

\begin{theorem}
\label{thm:both-quotas}
    For every multi-level apportionment instance, there exists a seat allocation that fulfills both upper and lower quota.
\end{theorem}

\Cref{thm:both-quotas} is an immediate consequence of Lemmas~\ref{lemma:reduction} and \ref{lemma:both-quotas-full-binary}, which we state and prove next.
Recall that a \emph{full binary tree} is a binary tree in which every node has either $0$ or $2$ children.

\begin{lemma}
\label{lemma:reduction}
Suppose that for every multi-level apportionment instance with a full binary tree, there exists a seat allocation that fulfills both upper and lower quota.
Then, the same holds for every multi-level apportionment instance (with an arbitrary tree).
\end{lemma}

\begin{lemma}
\label{lemma:both-quotas-full-binary}
For every multi-level apportionment instance with a full binary tree, there exists a seat allocation that fulfills both upper and lower quota.
\end{lemma}

\begin{proof}[Proof of \Cref{lemma:reduction}]
We claim that an arbitrary instance can be reduced to one with a full binary tree.
Given an arbitrary instance, we can apply the following transformation rules for each node $i$: 
\begin{itemize}
    \item If $|C(i)| = 0$ or $2$, do nothing;
    \item If $|C(i)| = 1$, merge $i$ with its only child, since the two nodes must receive the same number of seats;
    \item If $|C(i)| > 2$, introduce a new intermediate node $j$ to become the parent of $|C(i)| - 1$ children of~$i$, and make $i$ the parent of both $j$ and the remaining child~$i^*$ of $i$.
    Set $W_j = \sum_{c\in C(j)} W_c$, and replace $W_c$ by $\frac{W_c}{1 - W_{i^*}}$ for each $c\in C(j)$.
    Then, apply the same transformation rules to $j$.
\end{itemize}
An example of the reduction is shown in \Cref{fig:tree-transform}.
Observe that all quota constraints between nodes in the original tree are still present in the new tree.
Therefore, a seat allocation that fulfills both quotas in the new tree does so in the original tree as well.
\end{proof}

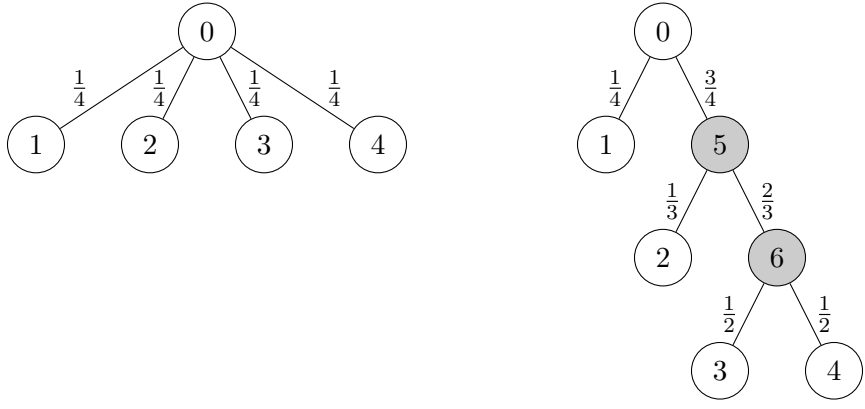
\begin{figure}[t]
\centering
\begin{tikzpicture}[
    node/.style={circle, draw, minimum size=0.75cm},
    edge from parent/.style={draw},
    level 1/.style={sibling distance=1.5cm, level distance=1.5cm},
    level 2/.style={sibling distance=1.5cm, level distance=1.5cm},
  ]

  % Left tree
  \node[node] (A) {$0$}
    child {node[node] {$1$} edge from parent node[left, xshift=-0.3cm] {$\frac{1}{4}$}}
    child {node[node] {$2$} edge from parent node[left] {$\frac{1}{4}$}}
    child {node[node] {$3$} edge from parent node[right] {$\frac{1}{4}$}}
    child {node[node] {$4$} edge from parent node[right, xshift=0.3cm] {$\frac{1}{4}$}};
    
  % Right tree
    \node[node] (B) at (6, 0) {$0$}
    child {node[node] {$1$} edge from parent node[left] {$\frac{1}{4}$}}
    child {node[node, fill=black!20] {$5$}
      child {node[node] {$2$} edge from parent node[left] {$\frac{1}{3}$}}
      child { node[node, fill=black!20] {$6$} 
        child {node[node] {$3$} edge from parent node[left] {$\frac{1}{2}$}}
        child {node[node] {$4$} edge from parent node[right] {$\frac{1}{2}$}}
        edge from parent node[right] {$\frac{2}{3}$}}
      edge from parent node[right] {$\frac{3}{4}$}};
\end{tikzpicture}

    \caption{Illustration of reducing an arbitrary tree to a full binary tree in the proof of \Cref{lemma:reduction}.
    Shaded nodes are the introduced intermediate nodes.}
    \label{fig:tree-transform}
\end{figure}

\begin{proof}[Proof of \Cref{lemma:both-quotas-full-binary}]
We proceed by induction, showing that if a satisfactory seat allocation can be found for all full binary trees with $n$ nodes, then the same is true for all full binary trees with $n+2$ nodes (note that the number of nodes in a full binary tree is always odd).
The base case $n = 1$ holds trivially.

Consider an arbitrary instance with $n+2 \ge 3$ nodes.
Let $i$ be a node such that both of its children, denoted by $x$ and $y$, are leaf nodes---such a node~$i$ always exists in a full binary tree when $n > 1$.
We have $W_x,W_y > 0$ and $W_x + W_y = 1$.
By the induction hypothesis, in the tree with $x$ and $y$ removed, there exists a seat allocation (with the same number of seats) that fulfills both quotas.

For any two nodes $c$ and $a$ such that $a \in A(c)$, let $e_{c,a} = \frac{R_c}{R_a} \cdot V_a$, and recall that $LQ_c = \max_{a \in A(c)} \lfloor e_{c,a} \rfloor$ and $UQ_c = \min_{a \in A(c)} \lceil e_{c,a} \rceil$. Fix $c\in\{x,y\}$, let $e_{\text{max}} := \max_{a \in A(c)} e_{c,a}$ and  $e_{\text{min}} := \min_{a \in A(c)} e_{c,a}$, and suppose that the maximum and minimum are achieved at ancestors $a_1$ and $a_2$ of $c$, respectively.
We will show that $e_{\text{max}} - e_{\text{min}} < 1$.
This is clear if $a_1 = a_2$, so assume that $a_1\ne a_2$.
If $a_1$ occupies a higher position in the tree than $a_2$, we have
\begin{align*}
0\le e_{\text{max}} - e_{\text{min}}
=  \frac{R_c}{R_{a_1}} \cdot V_{a_1} - \frac{R_c}{R_{a_2}} \cdot V_{a_2}  = \frac{R_c}{R_{a_2}} \left( \frac{R_{a_2}}{R_{a_1}} \cdot V_{a_1} - V_{a_2} \right).
\end{align*}
Notice that $\frac{R_c}{R_{a_2}} < 1$ since $a_2$ is an ancestor of $c$, and $\frac{R_{a_2}}{R_{a_1}} \cdot V_{a_1} - V_{a_2} < 1$ since lower quota for $a_2$ with respect to $a_1$ is satisfied.
Hence, $e_{\text{max}} - e_{\text{min}} < 1$ in this case.
On the other hand, if $a_2$ occupies a higher position in the tree than $a_1$, we have
\begin{align*}
0\le e_{\text{max}} - e_{\text{min}}
=  \frac{R_c}{R_{a_1}} \cdot V_{a_1} - \frac{R_c}{R_{a_2}} \cdot V_{a_2}  = \frac{R_c}{R_{a_1}} \left(V_{a_1} - \frac{R_{a_1}}{R_{a_2}} \cdot V_{a_2} \right).
\end{align*}
Notice that $\frac{R_c}{R_{a_1}} < 1$ since $a_1$ is an ancestor of $c$, and $V_{a_1} - \frac{R_{a_1}}{R_{a_2}} \cdot V_{a_2} < 1$ since upper quota for $a_1$ with respect to $a_2$ is satisfied.
Hence, $e_{\text{max}} - e_{\text{min}} < 1$ always holds.
Since $LQ_c = \lfloor e_{\text{max}} \rfloor$ and $UQ_c = \lceil e_{\text{min}} \rceil$, it follows that $LQ_c \leq UQ_c$.
Thus, there exists a number of seats $V_c$ that can be allocated to node~$c$ so that both quotas are satisfied for $c$ with respect to all of its ancestors.

To complete the proof, it remains to show that we can choose $V_x$ and $V_y$ so that $V_i = V_x + V_y$.
Since $R_x + R_y = R_i$, it holds that $e_{x, a} + e_{y, a} = e_{i, a}$ for each $a\in A(i)\cup\{i\}$. 
Thus, 
\begin{align}
\lfloor e_{x, a} \rfloor + \lfloor e_{y, a} \rfloor
\leq \lfloor e_{i, a} \rfloor \leq V_i \leq \lceil e_{i, a} \rceil \leq \lceil e_{x, a} \rceil + \lceil e_{y, a} \rceil \label{eq:lv-uv-per-anc}
\end{align}
for each $a\in A(i)\cup\{i\}$.

Now, for any $c\in\{x,y\}$ and any $a_1, a_2 \in A(c)$ (note that $A(c) = A(i)\cup\{i\}$), we know that $\frac{e_{x, a_1}}{e_{y, a_1}} = \frac{e_{x, a_2}}{e_{y, a_2}} = \frac{R_x}{R_y}$, which means that $e_{x, a_1} \leq e_{x, a_2}$ implies $e_{y, a_1} \leq e_{y, a_2}$. 
From the definition $LQ_c = \max_{a \in A(c)} \lfloor e_{c, a} \rfloor$, we have
    \begin{align*}
    LQ_x + LQ_y &= \max_{a_1 \in A(c)} \lfloor e_{x, a_1} \rfloor + \max_{a_2 \in A(c)} \lfloor e_{y, a_2} \rfloor
    = \max_{a \in A(c)} (\lfloor e_{x, a} \rfloor + \lfloor e_{y, a} \rfloor),
    \end{align*}
where the second equality holds by the previous sentence.
Similarly, $UQ_x + UQ_y = \min_{a\in A(c)} (\lceil e_{x, a} \rceil + \lceil e_{y, a} \rceil)$.
Combining these with \eqref{eq:lv-uv-per-anc}, we find that 
    \begin{align*}
        LQ_x + LQ_y \leq V_i \leq UQ_x + UQ_y.
    \end{align*}
    Together with the previously established claims that $LQ_x \leq UQ_x$ and $LQ_y \leq UQ_y$, this implies that we can choose $V_x$ and $V_y$ so that $V_i = V_x + V_y$, as desired.
\end{proof}

\section{Evaluating Adams and Jefferson}
\label{sec:adams-jefferson}

While \Cref{thm:both-quotas} ensures that both upper and lower quota can be fulfilled in any multi-level apportionment instance, it does not give rise to a method that satisfies house monotonicity.
A simple way to create such a method is to take an arbitrary house-monotone single-level method and run it on a multi-level instance one level at a time, starting at the root node and going down the tree---it is clear that if the total number of seats increases, no node receives fewer seats as a result.
In this section, we show that if we take Adams' method (resp., Jefferson's method),\footnote{The iterative form of these methods that we use here is described, e.g., by \citet[p.~100, Prop.~3.3]{BalinskiYo01}.} which satisfies upper quota (resp., lower quota) in single-level apportionment, the corresponding quota compliance is preserved in the multi-level setting.

We start with Adams' method, whose multi-level version is displayed as \Cref{algo:adams}.
The algorithm iteratively allocates one seat at a time based on the existing allocation.
We first prove an auxiliary lemma.

\begin{algorithm}[t]
    \caption{Multi-Level Adams' Method}
    \label{algo:adams}

    \DontPrintSemicolon
    \SetKwInOut{Input}{Input}
    \SetKwInOut{Output}{Output}
    \SetArgSty{textnormal}
    
    \Input {List $V^h$ representing the allocation for $h$ seats}
    \Output {List $V^{h+1}$ representing the allocation for $h+1$ seats}

    $V^{h+1} \gets V^h$ \hspace{5mm} \tcp{initialize $V^{h+1}$ to be $V^h$}
    $i \gets 0$ \hspace{5mm} \tcp{root node}
    $V^{h+1}_i \gets V^{h}_i + 1$ \hspace{5mm} \tcp{root node always receives the new seat}
    \While {$i \text{ is not a leaf node}$} {
        Let $c$ be a child of $i$ with the smallest $\frac{V^h_c}{R_c}$, breaking ties arbitrarily\;
        $V^{h+1}_c \gets V^{h}_c + 1$\;
        $i \gets c$\;
    }
\end{algorithm}

\begin{lemma}
\label{lemma:adams}
Let $h\ge 1$.
In \Cref{algo:adams}, suppose that the nodes that receive the $h$-th seat are $a_0,a_1,\dots,a_k$ from top to bottom, where $a_0 = 0$.
Then, $\frac{V^h_{a_{j-1}}}{R_{a_{j-1}}} \ge \frac{V^h_{a_{j}}}{R_{a_{j}}}$ for each $j\in \{1,2,\dots,k\}$.
\end{lemma}

\begin{proof}
Fix any $j\in\{1,2,\dots,k\}$, and assume for contradiction that $\frac{V^h_{a_{j-1}}}{R_{a_{j-1}}} < \frac{V^h_{d}}{R_{d}}$ for all $d\in C(a_{j-1})$.
Then, we have 
\begin{align*}
V^h_{a_{j-1}} 
= \sum_{d\in C(a_{j-1})} V^h_d 
> \sum_{d\in C(a_{j-1})} R_d\cdot \frac{V^h_{a_{j-1}}}{R_{a_{j-1}}}
= V^h_{a_{j-1}} \cdot \frac{\sum_{d\in C(a_{j-1})} R_d}{R_{a_{j-1}}}
= V^h_{a_{j-1}},
\end{align*}
a contradiction.
Hence, $\frac{V^h_{a_{j-1}}}{R_{a_{j-1}}} \ge \frac{V^h_{d}}{R_{d}}$ for some $d\in C(a_{j-1})$.
Since the algorithm chooses a node $a_j\in C(a_{j-1})$ with the smallest $\frac{V^h_{a_j}}{R_{a_j}}$, we have $\frac{V^h_{a_{j-1}}}{R_{a_{j-1}}} \ge \frac{V^h_{a_{j}}}{R_{a_{j}}}$.
It follows that $\frac{V^h_{a_{0}}}{R_{a_{0}}} \ge \frac{V^h_{a_{1}}}{R_{a_{1}}} \ge \dots \ge \frac{V^h_{a_{k}}}{R_{a_{k}}}$.
\end{proof}

\begin{theorem}
    \label{thm:adams-lower-quota}
    \Cref{algo:adams} satisfies upper quota and house monotonicity in multi-level apportionment.
\end{theorem}

\begin{proof}
House monotonicity is clear from the iterative nature of the algorithm.
We prove by induction on $h$ that, for any $h\ge 0$, the allocation of $h$ seats fulfills upper quota.
The base case $h = 0$ is trivial.
Assume that $h$ seats have been allocated for some $h\ge 0$, and consider the allocation of the $(h+1)$-th seat.
Let the nodes that receive this seat be $a_0,a_1,\dots,a_k$ from top to bottom, where $a_0 = 0$.
It follows from \Cref{lemma:adams} that $\frac{V^h_{a_{j-1}}}{R_{a_{j-1}}} \ge \frac{V^h_{a_{j}}}{R_{a_{j}}}$ for each $j\in \{1,2,\dots,k\}$.

By the induction hypothesis, after the $(h+1)$-th seat is allocated, the only potential upper quota violations are those for some node $a_j$ with respect to its ancestor $a_\ell$, where $\ell < j$.
From the previous paragraph, we know that $\frac{V^h_{a_{\ell}}}{R_{a_{\ell}}} \ge \frac{V^h_{a_{j}}}{R_{a_{j}}}$.
Since $V^{h+1}_{a_\ell} = V^h_{a_\ell} + 1$ and $V^{h+1}_{a_j} = V^h_{a_j} + 1$, we have $\frac{V^{h+1}_{a_{\ell}} - 1}{R_{a_{\ell}}} \ge \frac{V^{h+1}_{a_{j}} - 1}{R_{a_{j}}}$, which implies that $V^{h+1}_{a_j} \le \frac{R_{a_j}}{R_{a_\ell}}\cdot V^{h+1}_{a_\ell} + 1 - \frac{R_{a_j}}{R_{a_\ell}} < \frac{R_{a_j}}{R_{a_\ell}}\cdot V^{h+1}_{a_\ell} + 1$.
It follows that $V^{h+1}_{a_j} \le \left\lceil \frac{R_{a_j}}{R_{a_\ell}}\cdot V^{h+1}_{a_\ell} \right\rceil$, which means that upper quota for $a_j$ with respect to $a_\ell$ is satisfied upon the allocation of the $(h+1)$-th seat, completing the proof.
\end{proof}

We now turn our attention to Jefferson's method, whose multi-level version is described as \Cref{algo:jefferson}.
The only difference from \Cref{algo:adams} is in \Cref{line:jefferson-ratio}, where the ratio $\frac{V_c^h+1}{R_c}$ is used instead of~$\frac{V_c^h}{R_c}$.
Again, we first establish an auxiliary lemma.

\begin{algorithm}[t]
    \caption{Multi-Level Jefferson's Method}
    \label{algo:jefferson}

    \DontPrintSemicolon
    \SetKwInOut{Input}{Input}
    \SetKwInOut{Output}{Output}
    \SetArgSty{textnormal}

    \Input {List $V^h$ representing the allocation for $h$ seats}
    \Output {List $V^{h+1}$ representing the allocation for $h+1$ seats}

    $V^{h+1} \gets V^h$ \hspace{5mm} \tcp{initialize $V^{h+1}$ to be $V^h$}
    $i \gets 0$ \hspace{5mm} \tcp{root node}
    $V^{h+1}_i \gets V^{h}_i + 1$ \hspace{5mm} \tcp{root node always receives the new seat}
    \While {$i \text{ is not a leaf node}$} {
        Let $c$ be a child of $i$ with the smallest $\frac{V^h_c + 1}{R_c}$, breaking ties arbitrarily\; \label{line:jefferson-ratio}
        $V^{h+1}_c \gets V^{h}_c + 1$\;
        $i \gets c$\;
    }
\end{algorithm}

\begin{lemma}
\label{lemma:jefferson-lq}
Let $h\ge 0$.
After $h$ seats have been allocated by \Cref{algo:jefferson}, for any two nodes $p$ and $c$ such that $c\in C(p)$, it holds that $\frac{V^h_c + 1}{R_c} \geq \frac{V^h_p + 1}{R_p}$.
\end{lemma}

\begin{proof}
First, we claim that whenever $c$ has just received a seat---say, the $g$-th seat---it holds that $\frac{V^g_c}{R_c} \ge \frac{V^g_p}{R_p}$.
Suppose for contradiction that $\frac{V^g_c}{R_c} < \frac{V^g_p}{R_p}$.
We know that
\begin{align*}
\frac{V^g_p}{R_p} 
= \sum_{c'\in C(p)} \frac{V^g_{c'}}{R_p}
= \sum_{c'\in C(p)} \frac{V^g_{c'}}{R_{c'}}\cdot \frac{R_{c'}}{R_p}
= \sum_{c'\in C(p)} \frac{V^g_{c'}}{R_{c'}}\cdot W_{c'}.
\end{align*}
Since $\sum_{c'\in C(p)} W_{c'} = 1$ and $\frac{V^g_c}{R_c} < \frac{V^g_p}{R_p}$, there exists $c^*\in C(p)$ such that $\frac{V^g_{c^*}}{R_{c^*}} > \frac{V^g_p}{R_p} > \frac{V^g_c}{R_c}$.
Hence, just before $c$ received this seat, $\frac{V^{g-1}_{c^*}}{R_{c^*}} > \frac{V^{g-1}_c+1}{R_c}$.
However, this means that the last seat allocated to $c^*$ before this point should have been allocated to $c$ instead, a contradiction that establishes the claim.

We now proceed to prove the lemma.
Assume for contradiction that for some $h$, after $h$ seats have been allocated, $\frac{V^h_c + 1}{R_c} < \frac{V^h_p + 1}{R_p}$.
Let $h' > h$ be the round in which $c$ receives the next seat; it holds that $\frac{V^{h'-1}_c + 1}{R_c} = \frac{V^{h}_c + 1}{R_c} < \frac{V^h_p + 1}{R_p} \le \frac{V^{h'-1}_p + 1}{R_p}$.
Hence, after the $h'$-th seat has been allocated to $c$, we have $\frac{V^{h'}_c}{R_c} < \frac{V^{h'}_p}{R_p}$.
However, this contradicts the claim in the previous paragraph, completing the proof.
\end{proof}

\begin{theorem}
\label{thm:jefferson-lower-quota}
    \Cref{algo:jefferson} satisfies lower quota and house monotonicity in multi-level apportionment.
\end{theorem}

\begin{proof}
House monotonicity is clear from the iterative nature of the algorithm.
For lower quota, consider any pair of nodes $a,b$ such that $a$ is an ancestor of $b$.
By applying \Cref{lemma:jefferson-lq} repeatedly along the path between $a$ and $b$, we have that at any point during the execution of the algorithm, it holds that $\frac{V_b+1}{R_b} \ge \frac{V_a+1}{R_a}$.
This means that $V_b \ge \frac{R_b}{R_a}\cdot V_a - 1 + \frac{R_b}{R_a} > \frac{R_b}{R_a}\cdot V_a - 1$, so $V_b \ge \left\lfloor \frac{R_b}{R_a}\cdot V_a \right\rfloor$.
Hence, lower quota for $b$ with respect to $a$ is satisfied.
\end{proof}

\section{Applying the Quota Method}
\label{sec:quota-method}

Given the results in \Cref{sec:adams-jefferson}, a natural next approach is to take an iterative method that satisfies both upper and lower quota in single-level apportionment, and apply it level-by-level in the multi-level setting.
A well-known method with this property is the quota method \citep[Sec.~5]{BalinskiYo75}, whose multi-level version is presented as \Cref{algo:quota1}.
The algorithm is the same as multi-level Jefferson (\Cref{algo:jefferson}), except that there is a check against upper quota violations with respect to the parent of each node (\Cref{line:upper-quota-condition}).

We first show that the multi-level quota method satisfies lower quota.
The proof proceeds via a similar lemma as for multi-level Jefferson (\Cref{lemma:jefferson-lq}), but proving the lemma is now more involved.

\begin{lemma}
\label{lemma:quota1-lq}
Let $h\ge 0$.
After $h$ seats have been allocated by \Cref{algo:quota1}, for any two nodes $p$ and $c$ such that $c\in C(p)$, it holds that $\frac{V^h_c + 1}{R_c} \geq \frac{V^h_p + 1}{R_p}$.
\end{lemma}

\begin{proof}
It suffices to prove that whenever $c$ has just received a seat, it holds that $\frac{V_c}{R_c} \ge \frac{V_p}{R_p}$. 
Once we have this claim, we can finish the proof of the lemma as in the last paragraph of the proof of \Cref{lemma:jefferson-lq}.

Suppose for contradiction that for some $h$, after the $h$-th seat has been allocated to $c$, it holds that $\frac{V^h_c}{R_c} < \frac{V^h_p}{R_p}$.
As in the proof of \Cref{lemma:jefferson-lq}, this means that there exists $c^*\in C(p)$ such that $\frac{V^h_{c^*}}{R_{c^*}} > \frac{V^h_p}{R_p} > \frac{V^h_c}{R_c}$.
Hence, the set 
$I := \{i\in C(p)\mid \frac{V^h_i}{R_i} > \frac{V^h_p}{R_p} \}$ is non-empty and $c\not\in I$.
Let $g\ge 0$ be the smallest integer such that $V^g_i = V^h_i$ for all $i\in I$.
Since $c$ receives the $h$-th seat, we have $g < h$.
Also, $g = 0$ would imply that $V^g_i = V^h_i = 0$ for all $i\in I$, contradicting the definition of $I$, so $g > 0$.
Let $J := \{j\in C(p)\mid V^h_j > V^g_j\}$; by definition, we know that $c\in J$ and $J\cap I = \emptyset$.
For each $j\in J$, since $j\not\in I$, we have $\frac{V^h_j}{R_j} \le \frac{V^h_p}{R_p}$.
Also, by definition of~$g$, there exists a node $i^*\in I$ that receives the $g$-th seat.
Thus, for each $j\in J$, we have
$\frac{V^g_{i^*}}{R_{i^*}} = \frac{V^h_{i^*}}{R_{i^*}} > \frac{V^h_p}{R_p} \ge \frac{V^h_j}{R_j} \ge \frac{V^g_j + 1}{R_j}$.
Since $j$ does not receive the $g$-th seat, it must be that receiving the $g$-th seat would violate the upper quota check for $j$, and so $\frac{V^g_j}{R_j} \ge \frac{V^g_p}{R_p}$.

\begin{algorithm}[t]
    \caption{Multi-Level Quota Method}
    \label{algo:quota1}
    
    \DontPrintSemicolon
    \SetKwInOut{Input}{Input}
    \SetKwInOut{Output}{Output}
    \SetArgSty{textnormal}

    \Input {List $V^h$ representing the allocation for $h$ seats}
    \Output {List $V^{h+1}$ representing the allocation for $h+1$ seats}

    $V^{h+1} \gets V^h$ \hspace{5mm} \tcp{initialize $V^{h+1}$ to be $V^h$}
    $i \gets 0$ \hspace{5mm} \tcp{root node}
    $V^{h+1}_i \gets V^{h}_i + 1$ \hspace{5mm} \tcp{root node always receives the new seat}
    \While {$i \text{ is not a leaf node}$} {
        Let $c$ be a child of $i$ with the smallest $\frac{V^h_c + 1}{R_c}$ subject to $\frac{V^h_c}{W_c} < V^h_i + 1$, breaking ties arbitrarily\; \label{line:upper-quota-condition}
        $V^{h+1}_c \gets V^{h}_c + 1$\;
        $i \gets c$\;
    }
\end{algorithm}

Summarizing what we have so far, we know that $\frac{V^h_j}{R_j} \leq \frac{V^h_p}{R_p}$ and $\frac{V^g_j}{R_j} \geq \frac{V^g_p}{R_p}$ for all $j\in J$.
Combining them yields $\frac{V^h_j - V^g_j}{R_j} \leq \frac{V^h_p - V^g_p}{R_p}$ for all $j \in J$, or equivalently, $V^h_j - V^g_j \le \frac{R_j}{R_p}(V^h_p - V^g_p)$.
Summing this up across all $j\in J$, we get $\sum_{j\in J}(V^h_j - V^g_j) \le \frac{\sum_{j\in J}R_j}{R_p}\cdot (V^h_p - V^g_p)$.
Since $\sum_{j\in J}(V^h_j - V^g_j) = V^h_p - V^g_p$ by definition of $J$, and $V^h_p - V^g_p > 0$ because $c$ receives the $h$-th seat, this simplifies to $R_p\le \sum_{j\in J}R_j$.
However, since $I$ is non-empty, we must have $\sum_{j\in J}R_j < R_p$, a contradiction.
\end{proof}

With \Cref{lemma:quota1-lq} in hand, we can establish lower quota in exactly the same way as for multi-level Jefferson (\Cref{thm:jefferson-lower-quota}).

\begin{theorem}
\label{thm:quota1-lower-quota}
    \Cref{algo:quota1} satisfies lower quota and house monotonicity in multi-level apportionment.
\end{theorem}

Unfortunately, even though the algorithm has a check against upper quota violations with respect to the parent, this is insufficient to guarantee upper quota in the multi-level setting.

\begin{proposition}
\label{prop:quota1-upper-quota}
\Cref{algo:quota1} fails upper quota in multi-level apportionment.
\end{proposition}

\begin{proof}
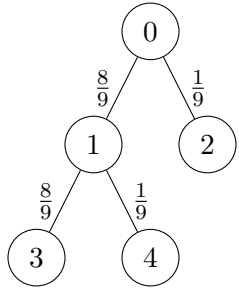
\begin{figure}[t]
\centering
\begin{tikzpicture}[
    node/.style={circle, draw, minimum size=0.75cm},
    edge from parent/.style={draw},
    level 1/.style={sibling distance=1.5cm, level distance=1.5cm},
    level 2/.style={sibling distance=1.5cm, level distance=1.5cm},
  ]

  % Tree structure
  \node[node] {$0$}
    child {
      node[node] {$1$}
      child { 
        node[node] {$3$}
        edge from parent node[left] {$\frac{8}{9}$}
      }
      child {
        node[node] {$4$}
        edge from parent node[right] {$\frac{1}{9}$}
      }
      edge from parent node[left] {$\frac{8}{9}$}
    }
    child {
      node[node] {$2$}
      edge from parent node[right] {$\frac{1}{9}$}
    };
\end{tikzpicture}
\caption{Example tree structure for which \Cref{algo:quota1} fails upper quota.}
\label{fig:quota1-uq-tree}
\end{figure}

Consider the instance where the tree structure is as shown in \Cref{fig:quota1-uq-tree}, and assume that there are $h = 5$ seats to be allocated.
Since $\frac{4+1}{8/9} < \frac{0+1}{1/9}$ and $\frac{4}{8/9} < 4+1$, the algorithm allocates all five seats to node~$1$.
By the same reasoning, the algorithm allocates all five seats to node~$3$.
However, since $\lceil 5\cdot\frac{8}{9}\cdot\frac{8}{9}\rceil = 4$, upper quota for node~$3$ with respect to node~$0$ is violated.
\end{proof}

\begin{algorithm}[h]
    \caption{Upper-Compliant Multi-Level Quota Method}
    \label{algo:quota2}

    \DontPrintSemicolon
    \SetKwInOut{Input}{Input}
    \SetKwInOut{Output}{Output}
    \SetArgSty{textnormal}

    \Input {List $V^h$ representing the allocation for $h$ seats}
    \Output {List $V^{h+1}$ representing the allocation for $h+1$ seats}

    $V^{h+1} \gets V^h$ \hspace{5mm} \tcp{initialize $V^{h+1}$ to be $V^h$}
    $i \gets 0$ \hspace{5mm} \tcp{root node}
    $V^{h+1}_i \gets V^{h}_i + 1$ \hspace{5mm} \tcp{root node always receives the new seat}
    \tcc{Threshold value to prevent upper quota violation. When considering a node~$c$, we have $t = \min_{a \in A(c)} \frac{V^h_a + 1}{R_a}$. }
    $t \gets V^h_0 + 1$ \hspace{5mm} \tcp{Note that, by definition, $R_0 = 1$.}
    \While {$i \text{ is not a leaf node}$} {
        Let $c$ be a child of $i$ with the smallest $\frac{V^h_c + 1}{R_c}$ subject to $\frac{V^h_c}{R_c} < t$, breaking ties arbitrarily\;
        $V^{h+1}_c \gets V^{h}_c + 1$\;
        $i \gets c$\;
        $t \gets \min (t, \frac{V^h_{c} + 1}{R_{c}})$\;
    }
\end{algorithm}

To fix the failure in \Cref{prop:quota1-upper-quota}, one could try to prevent upper quota violations of a node not only with respect to its parent, but also with respect to all of its ancestors.
The resulting algorithm is shown as \Cref{algo:quota2}.
In order for a node~$c$ to be eligible to receive the next seat, the algorithm requires that $\frac{V_c}{R_c} < \frac{V_a+1}{R_a}$ for all $a\in A(c)$.
Unlike the previous three algorithms, it is not trivial to see that \Cref{algo:quota2} is well-defined---that is, if a non-leaf node~$i$ receives a seat, then at least one of its children is also eligible to receive the seat.
We establish this fact in the following proposition.

\begin{proposition}
\Cref{algo:quota2} is well-defined.
\end{proposition}

\begin{proof}
Suppose for contradiction that a node~$i$ receives a seat but none of its children is eligible to receive the seat.
This means that $\frac{V_c}{R_c} \ge \min_{a\in A(i)\cup\{i\}} \frac{V_a+1}{R_a}$ for all $c\in C(i)$.
If $\min_{a\in A(i)\cup\{i\}} \frac{V_a+1}{R_a} = \frac{V_i+1}{R_i}$, we would have 
\begin{align}
\frac{V_i}{R_i} 
= \sum_{c\in C(i)}\frac{V_c}{R_i} 
= \sum_{c\in C(i)}\frac{V_c}{R_c}\cdot W_c 
\ge \frac{V_i+1}{R_i}\cdot \sum_{c\in C(i)}W_c
= \frac{V_i+1}{R_i},
\label{eq:quota2}
\end{align}
a contradiction.
Hence, $\frac{V_c}{R_c} \ge \min_{a\in A(i)} \frac{V_a+1}{R_a}$ for all $c\in C(i)$.
Using the same reasoning as in~(\ref{eq:quota2}), we get $\frac{V_i}{R_i} \ge \min_{a\in A(i)} \frac{V_a+1}{R_a}$.
However, this means that $i$ should not have received the seat, a contradiction.
\end{proof}

Since \Cref{algo:quota2} prevents all possible upper quota violations when assigning each seat, upper quota is satisfied.
In addition, house monotonicity is fulfilled due to the iterative nature of the algorithm.

\begin{theorem}
\label{thm:quota2-upper-quota}
    \Cref{algo:quota2} satisfies upper quota and house monotonicity in multi-level apportionment.
\end{theorem}

However, the upper quota check across multiple levels can cause lower quota to be violated.

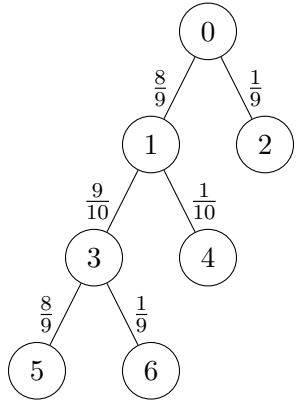
\begin{figure}[t]
\centering
\begin{tikzpicture}[
    node/.style={circle, draw, minimum size=0.75cm},
    edge from parent/.style={draw},
    level 1/.style={sibling distance=1.5cm, level distance=1.5cm},
    level 2/.style={sibling distance=1.5cm, level distance=1.5cm},
  ]

  % Tree structure
  \node[node] {$0$}
    child {
      node[node] {$1$}
      child { 
        node[node] {$3$}
        child { 
            node[node] {$5$}
            edge from parent node[left] {$\frac{8}{9}$}
          }
        child { 
            node[node] {$6$}
            edge from parent node[right] {$\frac{1}{9}$}
        }
        edge from parent node[left] {$\frac{9}{10}$}
      }
      child { 
        node[node] {$4$}
        edge from parent node[right] {$\frac{1}{10}$}
      }
      edge from parent node[left] {$\frac{8}{9}$}
    }
    child {
      node[node] {$2$}
      edge from parent node[right] {$\frac{1}{9}$}
    };
\end{tikzpicture}
\caption{Example tree structure where \Cref{algo:quota2} fails lower quota.}
\label{fig:quota2-lq-tree}
\end{figure}

\begin{proposition}
\label{prop:quota2-lower-quota}
\Cref{algo:quota2} fails lower quota in multi-level apportionment.
\end{proposition}

\begin{proof}
Consider the instance where the tree structure is as shown in \Cref{fig:quota2-lq-tree}, and assume that there are $h = 5$ seats to be allocated.
As in the proof of \Cref{prop:quota1-upper-quota}, the algorithm allocates all five seats to node~$1$.
Without the upper quota check, node~$3$ would receive all five seats.
However, because $\frac{3}{(8/9)\cdot (9/10)} < 3+1$ while $\frac{4}{(8/9)\cdot (9/10)} \ge 4+1$, it only receives the first four seats.
Similarly, because $\frac{2}{(8/9)\cdot (9/10)\cdot (8/9)} < 2+1$ while 
$\frac{3}{(8/9) \cdot (9/10)\cdot (8/9)} \ge 3+1$, among the four seats that node~$3$ receives, only the first three are passed on to node~$5$.
Since $\lfloor 5\cdot\frac{9}{10}\cdot\frac{8}{9}\rfloor = 4$, this means that lower quota for node~$5$ with respect to node~$1$ is violated.
\end{proof}

In light of our results, an interesting question is whether there exists a rule that satisfies both upper and lower quota as well as house monotonicity in multi-level apportionment, as the quota method does in single-level apportionment \citep[Thm.~3]{BalinskiYo75}. 
This question is open even when we relax our quota notions to require compliance only relative to the root node (cf.~\Cref{foot:relaxation}).

\section{Experimental Evaluation}
\label{sec:experiment}

To complement our theoretical results, we experimentally evaluate Adams', Jefferson's, Quota, and Upper-Compliant (UC) Quota methods (Algorithms~\ref{algo:adams}--\ref{algo:quota2}, respectively). 
In particular, we run experiments with trees of heights 3, 4, 5, and 6, using the following two tree types:

\begin{itemize}
    \item Full binary tree. Following our claim in \Cref{lemma:reduction}, an arbitrary instance can be reduced to a full binary tree instance without relaxing any constraints, thereby making full binary trees natural for evaluation.
    We run tests with $n \in \{15, 31, 63, 127\}$, which correspond to perfect binary trees. However, since Jefferson's and Quota methods can be shown to be equivalent when a node has only two children, we also consider another tree type, as explained next.
    \item Full $4$-ary tree. To keep the value of $n$ relatively small, we generate trees in the following manner: if we list the nodes at each level ordered from left to right, then except for the last level, nodes with even indices ($0$-based) have four children while the remaining nodes have zero children. We run tests with $n \in \{29, 61, 125, 253\}$; these values arise because the number of nodes at each level are $1$, $4$, $8$, $16$, $32$, $64$, and $128$, respectively.
\end{itemize}

For each combination of tree type and $n$-value, we generate $100{,}000$ instances with random entitlements. We assign entitlements as follows: for each non-leaf node, we generate a random array with length equal to the number of its children, where each entry is an integer between $1$ and $10$ (inclusive). The entitlement of each child equals its corresponding value divided by the sum of the array. This keeps the ratio of entitlement between any two sibling nodes at most $1:10$ and the assigned entitlements relatively small, which helps avoid precision errors when working with floating-point numbers.

We run the four algorithms for each generated instance and analyze the allocation results at $h \in \{100, 500\}$ with respect to the following four metrics:

\begin{itemize}
    \item Lower quota violation rate. This is calculated by dividing the number of nodes that violate their lower quota (as defined in \Cref{sec:prelim}) by the value of $n$, then averaging across all generated instances. The results are presented in Tables~\ref{table:lq_vio_binary}--\ref{table:lq_vio_4ary}, with Jefferson's and Quota methods omitted because \Cref{thm:jefferson-lower-quota,thm:quota1-lower-quota} show that they always satisfy lower quota.
    \item Upper quota violation rate. This is calculated by dividing the number of nodes that violate their upper quota (as defined in \Cref{sec:prelim}) by the value of $n$, then averaging across all generated instances. The results are presented in Tables~\ref{table:uq_vio_binary}--\ref{table:uq_vio_4ary}, with Adams' and Upper-Compliant Quota methods omitted because \Cref{thm:adams-lower-quota,thm:quota2-upper-quota} show that they always satisfy upper quota.
    \item Average deviation. This is calculated by averaging the absolute difference between the number of allocated seats and the strict quota\footnote{The \emph{strict quota} of a node~$i$ is defined as $R_i \cdot h$, where $R_i$ denotes the entitlement of node $i$ relative to the root (as defined in \Cref{sec:prelim}).} across all nodes and generated instances. The results are presented in Tables~\ref{table:avg_dev_quota_binary}--\ref{table:avg_dev_quota_4ary}.
    \item Maximum deviation. This is calculated by taking the maximum absolute difference between the number of allocated seats and the strict quota across all nodes and generated instances. The results are presented in Tables~\ref{table:max_dev_quota_binary}--\ref{table:max_dev_quota_4ary}.
\end{itemize}
Our code is available at \url{https://github.com/stevenwjy/multi-level-apportionment}.

From the results, we make the following observations:

\begin{itemize}
    \item As discussed earlier, Jefferson's and Quota methods yield identical results on binary trees.
    \item The violations and deviations tend to increase with $n$ for perfect binary trees, but decrease with~$n$ for $4$-ary trees. We hypothesize that the latter occurs because our construction makes $4$-ary trees sparser as the height increases.
    \item As shown in Tables~\ref{table:lq_vio_binary}--\ref{table:lq_vio_4ary}, the lower quota violation rate of the Upper-Compliant Quota method is substantially smaller than that of Adams' method.
    \item As shown in \Cref{table:uq_vio_4ary}, the upper quota violation rate of the Quota method is smaller than that of Jefferson's method in $4$-ary trees.
    \item As shown in Tables~\ref{table:avg_dev_quota_binary}--\ref{table:max_dev_quota_4ary}, the Upper-Compliant Quota method generally exhibits smaller average and maximum deviations than the other three algorithms.
\end{itemize}

% Lower Quota Violations - Binary Tree
\begin{table}[htbp]
\centering
\begin{tabular}{|l|cc|cc|cc|cc|}
\hline
\multirow{2}{*}{Method} & \multicolumn{2}{c|}{$n=15$} & \multicolumn{2}{c|}{$n=31$} & \multicolumn{2}{c|}{$n=63$} & \multicolumn{2}{c|}{$n=127$} \\
& $h=100$ & $h=500$ & $h=100$ & $h=500$ & $h=100$ & $h=500$ & $h=100$ & $h=500$ \\
\hline
Adams & 1.2747 & 1.0580 & 1.4816 & 1.3939 & 1.6183 & 1.5440 & 1.7694 & 1.6135 \\
UC Quota & 0.0120 & 0.0087 & 0.0300 & 0.0310 & 0.0502 & 0.0513 & 0.0609 & 0.0655 \\
\hline
\end{tabular}
\caption{Average lower quota violation rate (\% of nodes) for binary tree.}
\label{table:lq_vio_binary}
\end{table}

% Lower Quota Violations - $4$-ary Tree
\begin{table}[htbp]
\centering
\begin{tabular}{|l|cc|cc|cc|cc|}
\hline
\multirow{2}{*}{Method} & \multicolumn{2}{c|}{$n=29$} & \multicolumn{2}{c|}{$n=61$} & \multicolumn{2}{c|}{$n=125$} & \multicolumn{2}{c|}{$n=253$} \\
& $h=100$ & $h=500$ & $h=100$ & $h=500$ & $h=100$ & $h=500$ & $h=100$ & $h=500$ \\
\hline
Adams & 1.6681 & 1.4281 & 1.5177 & 1.5740 & 0.7814 & 1.4439 & 0.3863 & 0.8205 \\
UC Quota & 0.0001 & 0.0001 & 0.0000 & 0.0001 & 0.0000 & 0.0001 & 0.0000 & 0.0001 \\
\hline
\end{tabular}
\caption{Average lower quota violation rate (\% of nodes) for $4$-ary trees. Note that while some values for the Upper-Compliant Quota method appear as $0.0000$, they are non-zero but rounded to four decimal places.}
\label{table:lq_vio_4ary}
\end{table}

% Upper Quota Violations - Binary Tree
\begin{table}[htbp]
\centering
\begin{tabular}{|l|cc|cc|cc|cc|}
\hline
\multirow{2}{*}{Method} & \multicolumn{2}{c|}{$n=15$} & \multicolumn{2}{c|}{$n=31$} & \multicolumn{2}{c|}{$n=63$} & \multicolumn{2}{c|}{$n=127$} \\
& $h=100$ & $h=500$ & $h=100$ & $h=500$ & $h=100$ & $h=500$ & $h=100$ & $h=500$ \\
\hline
Jefferson & 0.9587 & 1.1367 & 1.3113 & 1.4119 & 1.4798 & 1.5552 & 1.4998 & 1.6120 \\
Quota & 0.9587 & 1.1367 & 1.3113 & 1.4119 & 1.4798 & 1.5552 & 1.4998 & 1.6120 \\
\hline
\end{tabular}
\caption{Average upper quota violation rate (\% of nodes) for binary tree.}
\label{table:uq_vio_binary}
\end{table}

% Upper Quota Violations - $4$-ary Tree
\begin{table}[htbp]
\centering
\begin{tabular}{|l|cc|cc|cc|cc|}
\hline
\multirow{2}{*}{Method} & \multicolumn{2}{c|}{$n=29$} & \multicolumn{2}{c|}{$n=61$} & \multicolumn{2}{c|}{$n=125$} & \multicolumn{2}{c|}{$n=253$} \\
& $h=100$ & $h=500$ & $h=100$ & $h=500$ & $h=100$ & $h=500$ & $h=100$ & $h=500$ \\
\hline
Jefferson & 1.5021 & 1.6556 & 1.1891 & 1.5936 & 0.6871 & 1.2565 & 0.3452 & 0.7668 \\
Quota & 0.4949 & 0.4806 & 0.4916 & 0.5660 & 0.2994 & 0.5209 & 0.1502 & 0.3387 \\
\hline
\end{tabular}
\caption{Average upper quota violation rate (\% of nodes) for $4$-ary tree.}
\label{table:uq_vio_4ary}
\end{table}

% Average Deviation - Binary Tree
\begin{table}[htbp]
\centering
\begin{tabular}{|l|cc|cc|cc|cc|}
\hline
\multirow{2}{*}{Method} & \multicolumn{2}{c|}{$n=15$} & \multicolumn{2}{c|}{$n=31$} & \multicolumn{2}{c|}{$n=63$} & \multicolumn{2}{c|}{$n=127$} \\
& $h=100$ & $h=500$ & $h=100$ & $h=500$ & $h=100$ & $h=500$ & $h=100$ & $h=500$ \\
\hline
Adams & 0.3009 & 0.2857 & 0.3251 & 0.3155 & 0.3457 & 0.3315 & 0.3683 & 0.3427 \\
Jefferson & 0.2850 & 0.2919 & 0.3156 & 0.3191 & 0.3296 & 0.3331 & 0.3288 & 0.3397 \\
Quota & 0.2850 & 0.2919 & 0.3156 & 0.3191 & 0.3296 & 0.3331 & 0.3288 & 0.3397 \\
UC Quota & 0.2628 & 0.2673 & 0.2846 & 0.2864 & 0.2955 & 0.2955 & 0.2983 & 0.3001 \\
\hline
\end{tabular}
\caption{Average deviation from strict quota for binary tree.}
\label{table:avg_dev_quota_binary}
\end{table}

% Average Deviation - $4$-ary Tree
\begin{table}[htbp]
\centering
\begin{tabular}{|l|cc|cc|cc|cc|}
\hline
\multirow{2}{*}{Method} & \multicolumn{2}{c|}{$n=29$} & \multicolumn{2}{c|}{$n=61$} & \multicolumn{2}{c|}{$n=125$} & \multicolumn{2}{c|}{$n=253$} \\
& $h=100$ & $h=500$ & $h=100$ & $h=500$ & $h=100$ & $h=500$ & $h=100$ & $h=500$ \\
\hline
Adams & 0.3468 & 0.3227 & 0.3581 & 0.3524 & 0.2757 & 0.3553 & 0.1807 & 0.2824 \\
Jefferson & 0.3216 & 0.3277 & 0.2896 & 0.3309 & 0.2078 & 0.2949 & 0.1249 & 0.2163 \\
Quota & 0.3083 & 0.3116 & 0.2815 & 0.3170 & 0.2038 & 0.2860 & 0.1230 & 0.2117 \\
UC Quota & 0.2945 & 0.2967 & 0.2712 & 0.3002 & 0.1991 & 0.2739 & 0.1208 & 0.2057 \\
\hline
\end{tabular}
\caption{Average deviation from strict quota for $4$-ary tree.}
\label{table:avg_dev_quota_4ary}
\end{table}

% Maximum Deviation - Binary Tree
\begin{table}[htbp]
\centering
\begin{tabular}{|l|cc|cc|cc|cc|}
\hline
\multirow{2}{*}{Method} & \multicolumn{2}{c|}{$n=15$} & \multicolumn{2}{c|}{$n=31$} & \multicolumn{2}{c|}{$n=63$} & \multicolumn{2}{c|}{$n=127$} \\
& $h=100$ & $h=500$ & $h=100$ & $h=500$ & $h=100$ & $h=500$ & $h=100$ & $h=500$ \\
\hline
Adams & 0.7713 & 0.7438 & 0.9306 & 0.9165 & 1.0777 & 1.0638 & 1.2304 & 1.2050 \\
Jefferson & 0.7402 & 0.7575 & 0.9118 & 0.9241 & 1.0566 & 1.0685 & 1.1930 & 1.2060 \\
Quota & 0.7402 & 0.7575 & 0.9118 & 0.9241 & 1.0566 & 1.0685 & 1.1930 & 1.2060 \\
UC Quota & 0.6555 & 0.6665 & 0.7589 & 0.7677 & 0.8279 & 0.8364 & 0.8803 & 0.8900 \\
\hline
\end{tabular}
\caption{Maximum deviation from strict quota for binary tree.}
\label{table:max_dev_quota_binary}
\end{table}

% Maximum Deviation - $4$-ary Tree
\begin{table}[htbp]
\centering
\begin{tabular}{|l|cc|cc|cc|cc|}
\hline
\multirow{2}{*}{Method} & \multicolumn{2}{c|}{$n=29$} & \multicolumn{2}{c|}{$n=61$} & \multicolumn{2}{c|}{$n=125$} & \multicolumn{2}{c|}{$n=253$} \\
& $h=100$ & $h=500$ & $h=100$ & $h=500$ & $h=100$ & $h=500$ & $h=100$ & $h=500$ \\
\hline
Adams & 0.9832 & 0.9447 & 1.1335 & 1.1171 & 1.1533 & 1.2537 & 1.1568 & 1.2801 \\
Jefferson & 0.9463 & 0.9701 & 1.0384 & 1.1066 & 1.0852 & 1.1990 & 1.1066 & 1.2412 \\
Quota & 0.8382 & 0.8512 & 0.9144 & 0.9356 & 0.9667 & 1.0028 & 0.9969 & 1.0452 \\
UC Quota & 0.7700 & 0.7877 & 0.8208 & 0.8308 & 0.8889 & 0.8663 & 0.9373 & 0.9154 \\
\hline
\end{tabular}
\caption{Maximum deviation from strict quota for $4$-ary tree.}
\label{table:max_dev_quota_4ary}
\end{table}

\FloatBarrier

\normalsize
\subsection*{Acknowledgments}

This work is partially supported by the Singapore Ministry of Education under grant number MOE-T2EP20221-0001 and by an NUS Start-up Grant. Part of this research was carried out while Ulrike Schmidt-Kraepelin was supported by the National Science Foundation under Grant No. DMS-1928930 and by the Alfred P. Sloan Foundation under grant G-2021-16778 while being in residence at the Simons Laufer Mathematical Sciences Institute (formerly MSRI) in Berkeley, California, during the Fall 2023 semester.
We would like to thank Ayumi Igarashi for helpful discussions and the anonymous reviewer for valuable feedback.

\normalsize

\bibliographystyle{plainnat}
\bibliography{main}

\end{document}